\pdfoutput=1
\documentclass[11pt,a4paper]{article}
\usepackage[a4paper,margin=2.55cm]{geometry}
\usepackage[T1]{fontenc}
\usepackage[utf8]{inputenc}
\usepackage{lmodern}
\usepackage{amsmath,amssymb,amsthm,mathtools}
\numberwithin{equation}{section}
\usepackage{enumitem}
\usepackage{needspace}
\usepackage{microtype}
\usepackage{booktabs}
\usepackage{array}
\usepackage{xcolor}
\usepackage[hidelinks]{hyperref}

\setlist{nosep,leftmargin=2em}

\newtheorem{theorem}{Theorem}[section]
\newtheorem{lemma}[theorem]{Lemma}
\newtheorem{proposition}[theorem]{Proposition}
\newtheorem{corollary}[theorem]{Corollary}
\theoremstyle{definition}

\theoremstyle{remark}

\newcommand{\OPT}{\operatorname{OPT}}
\newcommand{\Reach}{\operatorname{Reach}}

\newcommand{\pathset}{\mathcal{R}}
\newcommand{\Hreach}{\rightsquigarrow_H}

\newcommand{\R}{\mathbb{R}}

\title{Local Representatives and Shortest Completions for\\ Next-to-Shortest Paths in Directed Graphs}
\author{Shisheng Li}
\date{September 2026}

\begin{document}
\maketitle

\begin{abstract}
Given a directed graph with positive edge weights and two vertices $s,t$, a \emph{next-to-shortest} $s$--$t$ path is a shortest simple $s$--$t$ path among those whose length is strictly larger than the shortest-path distance. The problem was introduced by Lalgudi, Papaefthymiou and Potkonjak in 1996; it is NP-hard when zero-weight edges are allowed, and its complexity on positively weighted digraphs remained open for almost three decades until Chen, Wein and Zhang recently gave a polynomial-time algorithm running in $O(n^4m^3\log n)$ time. We give a substantially faster algorithm within their optimal-middle-segment framework.

The core idea is to split the problem into ``choosing a prefix'' and ``completing it.'' Given a prefix $P\colon s\to A$ that uses only shortest-path edges, delete the vertices already used by $P$, forbid leaving $A$ along shortest-path edges, and the best completion is a single ordinary shortest-path computation. The difficulty lies in choosing $P$: even for a fixed $A$, deciding whether some shortest prefix admits a completion is NP-complete.

We do not solve these fixed-$A$ subproblems one by one. Fix any globally optimal next-to-shortest path; its middle segment induces a \emph{boundary edge} $x\to c$ in the shortest-path DAG. For the correct triple $(A,B,x)$, the optimal path itself certifies $c$ as a feasible next hop, and we prove that every feasible next hop that is not earlier than $c$ in a topological order can be combined with the same middle segment into another globally optimal path. Consequently only the feasible next hop of maximum topological index needs to be kept for each triple, giving $O(n^3)$ representatives. A two-dimensional DAG dynamic program with a local reward generates all of these representatives simultaneously. The total running time is
\[
O\bigl(n^3(m+n\log n)\bigr),
\]
and $O(n^3m)$ on unweighted graphs. The correctness proof rests on an uncrossing lemma: whenever a candidate prefix is not compatible with the middle segment, the last intersection between a reference prefix and the candidate's partner suffix can be moved to a strictly earlier position along that suffix, which cannot go on forever.
\end{abstract}

\section{Introduction}

Let $G=(V,E,w)$ be a finite simple directed graph with $w\colon E\to\R_{>0}$, and let $s\ne t$ be two vertices; write $n=|V|$ and $m=|E|$. A \emph{next-to-shortest} $s$--$t$ path is a shortest path among all simple $s$--$t$ paths whose length is strictly larger than the shortest-path distance $d_G(s,t)$. The problem is also known as the \emph{strictly second shortest path} problem. It differs from the classical $k$ shortest simple paths problem of Yen~\cite{yen71} and Katoh, Ibaraki and Mine~\cite{kim82} in an essential way: the second entry of the $k$-shortest list can be another shortest path, and exponentially many shortest paths may have to be enumerated before the first strictly longer one appears.

\subsection{Background}

The problem was introduced by Lalgudi, Papaefthymiou and Potkonjak~\cite{lpp96,lpp00} in the context of block processing of data streams in hardware, where the strictly second shortest cycle or path determines the achievable delay. Shortly afterwards, Lalgudi and Papaefthymiou~\cite{lalgudi97} showed that the problem is NP-hard (its decision version NP-complete) on directed graphs with \emph{nonnegative} edge weights, by a reduction from the two vertex-disjoint paths problem in digraphs, which is NP-complete by the classical result of Fortune, Hopcroft and Wyllie~\cite{fhw80}. Their reduction relies on zero-weight edges. They also observed that the relaxation in which the path need not be simple (the \emph{walk} version) is easy. The status of the simple-path problem on directed graphs with \emph{positive} weights was left open.

On undirected graphs the problem turned out to be tractable, and a sequence of increasingly efficient algorithms followed. Krasikov and Noble~\cite{krasikov04} gave the first polynomial-time algorithm for positive weights, running in $O(n^3m)$ time; they also conjectured that the directed version is NP-complete. Li, Sun and Chen~\cite{lsc06} improved the running time to $O(n^3)$, Kao, Chang, Wang and Juan~\cite{kcwj11} to $O(n^2)$, and Wu~\cite{wu13} to $O(m+n\log n)$, which is linear once the distances from $s$ and to $t$ are known. The undirected problem with nonnegative weights was settled by Zhang and Nagamochi~\cite{zn12} and, with a linear-time algorithm after distance computation, by Wu, Guo and Wang~\cite{wgw12}. Special graph classes were also considered~\cite{mp06,bmp09}. Berger, Seymour and Spirkl~\cite{bss21} studied the related problem of finding an induced path that is not a shortest path.

For directed graphs with positive weights, Wu and Wang~\cite{wuwang15} derived structural properties of a next-to-shortest path in an arbitrary digraph and obtained an $O(n^3)$-time algorithm for planar digraphs; the general case remained open, and it was not even known how to find \emph{some} non-shortest simple $s$--$t$ path in polynomial time. The question was raised repeatedly in the literature (see the list in~\cite{cwz}), including in work on detours and on congested shortest paths~\cite{bcdf19,avwwx23,fgl23,hmps23,aw23}. Chen, Wein and Zhang~\cite{cwz} (henceforth CWZ) recently resolved it: they gave the first polynomial-time algorithm for the next-to-shortest path problem on positively weighted directed graphs, with running time $O(n^4m^3\log n)$. Their approach first reduces the input to an \emph{$(s,t)$-layered} digraph and then, guided by a structural ``Key Lemma'' about a fixed optimal solution, enumerates all choices of the two endpoints of the optimal middle segment together with a pair of forward edges on the same layer, and computes a pair of disjoint shortest paths through the chosen edges for each choice, using the linear-time algorithm of Tholey~\cite{tholey12} for two disjoint paths in a DAG (which builds on Perl and Shiloach~\cite{shiloach78}). The bound $O(n^4m^3\log n)$ is their Theorem~1.1. CWZ remark in their introduction that their focus was polynomiality rather than efficiency, that a bound of $O(n^4m^2\log n)$ seems attainable by avoiding edge subdivisions, and that obtaining a much faster algorithm is an interesting open question.

\subsection{Our result}

We give an algorithm for the next-to-shortest path problem on positively weighted directed graphs with running time
\[
O\bigl(n^3(m+n\log n)\bigr),
\]
which becomes $O(n^3m)$ on unweighted graphs. Compared with the $O(n^4m^3\log n)$ bound of CWZ this is an improvement by a factor of roughly $nm^2$, and the algorithm works directly on the input graph without any layering reduction. The working space is $O(n+m+h^2)$, where $h\le n$ is the number of vertices lying on shortest $s$--$t$ paths. The contribution is the combination of a new local sufficient condition for compatibility with the optimal middle segment, the selection of one representative prefix per triple of vertices, a one-sided completion step, and a direct proof on the input graph; the framework in which these are placed is due to CWZ, as we now explain.

\subsection{Relation to the approach of CWZ}\label{sec:cwz}

Our algorithm follows the optimal-middle-segment framework of CWZ. Their key idea is to fix a globally optimal next-to-shortest path, split it into a shortest prefix, a middle segment and a shortest suffix (in their layered setting, the segment between the first and the last back-edge), and to search not for the unknown middle segment but for a pair of vertex-disjoint shortest paths that is guaranteed to be compatible with it. Their Key Lemma shows that there is a pair of forward edges on the same layer such that \emph{every} pair of disjoint shortest paths through these two edges avoids the interior of the middle segment. The algorithm enumerates such pairs of edges, computes one disjoint pair for each of them, deletes the vertices of the pair except the two endpoints $A,B$ of the segment, and searches a shortest $A\to B$ path in the remainder. We keep this framework---fix the unknown middle segment and look for a shortest-path environment that preserves it---and change three things.

\emph{The local constraint.} Instead of prescribing an edge on each of the two paths, the Local Replacement Theorem (Theorem~\ref{thm:replacement}) constrains only the first path, and only at one vertex: $x\in P$, and the next hop of $P$ after $x$ is not earlier than a vertex $c$ in a topological order of the shortest-path DAG. The second path $Q\colon B\to t$ must still exist and be disjoint from $P$, but it is not required to pass through any prescribed edge, and the algorithm never enumerates $c$: for each triple $(A,B,x)$ it takes the feasible next hop of maximum topological index, which the optimal path itself certifies to be admissible. This is a new sufficient condition with its own proof, not a shortcut through the certificate of CWZ.

\emph{The completion.} CWZ fix both paths, delete their vertices and search a shortest $A\to B$ path. We use only the prefix $P$ to define the residual graph $G_P$---delete $P\setminus\{A\}$ and the tight out-edges of $A$---and search a shortest $A\to t$ path. The path $Q$ enters the selection of $P$ but not the completion: the suffix actually found need not pass through $B$, and need not coincide with the path $K\circ Q$ used in the proof. In both algorithms the inner step is one shortest-path computation; the difference lies in which paths are fixed, which vertices are deleted, and whether the target is $B$ or $t$.

\emph{The setting of the proof.} CWZ first reduce a general input to an $(s,t)$-straight graph and then to an $(s,t)$-layered graph. Our structural results are proved directly on the input graph: the shortest-path DAG $H$ provides the acyclic structure, and only the set $S$ of straight interior vertices of the middle segment enters the disjointness analysis, while non-straight vertices are left to the completion step. The correctness proof therefore does not invoke the layering reduction or the Key Lemma of CWZ as a black box---a statement about the dependencies of the proof, not about the origin of the method. The basic exchange argument, splicing a shorter non-shortest simple path and contradicting optimality, is already used by CWZ; we package it as middle-segment rigidity (Lemma~\ref{lem:rigidity}) and then derive the local replacement condition through barrier splicing and the strict descent of a last intersection (Lemma~\ref{lem:uncrossing}). The dynamic program is the standard two-token recursion on a DAG; what is specific to our use of it is the local reward and the sharing of one table across all $B$. CWZ also note that notions related to their back-edges appear in earlier work, including Wu and Wang~\cite{wuwang15}, and that the term \emph{$(s,t)$-straight} follows Berger, Seymour and Spirkl~\cite{bss21}; we use \emph{straight vertex} in the same sense.

\subsection{Technical overview}

Our starting point is a simple \emph{completion} operation. Let $H$ be the union of all shortest $s$--$t$ paths; it is a DAG. Given a prefix $P\colon s\to A$ inside $H$, delete the vertices of $P\setminus\{A\}$ and all shortest-path edges leaving $A$. A shortest $A\to t$ path in the residual graph, concatenated with $P$, is automatically simple and non-shortest. Hence, once the prefix is fixed, the best completion is a single shortest-path computation.

The difficulty is choosing the prefix. Let
\[
F(A)=\min\{C(P):P\text{ is an }s\to A\text{ path in }H\},
\]
where $C(P)$ is the length of the best completion of $P$. If a next-to-shortest path exists, then $\OPT=\min_A F(A)$. However, for a given $A$, already deciding whether $F(A)<\infty$ is NP-complete (Proposition~\ref{prop:fixedA-hard}, by a reduction from a disjoint-paths problem of Eilam-Tzoreff~\cite{eilam98}). So, unless $\mathrm{P}=\mathrm{NP}$, this decomposition does not yield a polynomial-time algorithm based on exactly evaluating every $F(A)$.

Our approach only requires a polynomial-size family of \emph{representative} prefixes that hits some global optimum. Fix a globally optimal next-to-shortest path $W^*=L_0KR_0$, where $L_0$ and $R_0$ are the longest tight prefix and suffix and $K$ is the middle segment; the first and last edges of $K$ are not shortest-path edges. The path $W^*$ is used only in the correctness proof; the algorithm never recovers $K$ and does not know the boundary vertex $c$ introduced below.

The proof starts from \emph{middle-segment rigidity}: if a proper subsegment of $K$ can be joined into a simple $s$--$t$ path by a tight prefix and a tight suffix, then it must consist of tight edges only. Rigidity yields a tight path that every tight $s\to A$ path must cross. If a candidate prefix cannot be combined with $K$, splicing tails on both sides of this barrier path moves the last intersection between a reference prefix and the candidate's partner suffix to a strictly earlier position. An \emph{uncrossing lemma} guarantees that this improvement can be repeated and that the last intersection always exists. A finite path does not admit an infinite strict descent, so the candidate prefix must in fact be compatible with $K$ (Theorem~\ref{thm:replacement}).

Algorithmically, the theorem says that for each triple $(A,B,x)$ of vertices of $H$ we only need one prefix: a prefix that continues from $x$ along the feasible next hop of maximum topological index. A two-dimensional DAG dynamic program with a local reward computes these representatives, and each representative is completed by one ordinary shortest-path computation. This yields the $O(n^3(m+n\log n))$ bound.

\subsection{Further related work}

The next-to-shortest path problem is intimately related to disjoint-paths problems, since a simple next-to-shortest path decomposes into segments that must be pairwise disjoint. The two vertex-disjoint paths problem is NP-complete in general digraphs~\cite{fhw80} but polynomial in DAGs~\cite{fhw80,shiloach78,tholey12}, and our dynamic program is the standard ``two tokens on a DAG'' technique of~\cite{fhw80} augmented with an additive reward. When each path is additionally required to be a shortest path, the problem is the $k$-disjoint shortest paths problem introduced by Eilam-Tzoreff~\cite{eilam98}, who solved the undirected case $k=2$ with positive lengths and posed the directed case; this was solved by B\'erczi and Kobayashi~\cite{bk17} for $k=2$, extended to nonnegative undirected lengths by Gottschau, Kaiser and Waldmann~\cite{gkw19}, and to every fixed $k$ on undirected graphs by Lochet~\cite{lochet21} and Bentert, Nichterlein, Renken and Zschoche~\cite{bnrz23}. Eilam-Tzoreff also showed that requiring only \emph{one} of two disjoint paths to be shortest gives an NP-complete problem, which is the source of hardness for our fixed-$A$ subproblem.

On \emph{unweighted} graphs, the existence version of the problem (find some non-shortest simple $s$--$t$ path) is the case $k=1$ of the \emph{$k$-longest detour} problem, which asks for an $s$--$t$ path with at least $d_G(s,t)+k$ edges. That problem is fixed-parameter tractable on undirected graphs~\cite{bcdf19,avwwx23} and on planar digraphs~\cite{fgl23,hmps23}; whether it is fixed-parameter tractable on general unweighted digraphs remains open, and CWZ suggest that their techniques may help. (With positive real weights this equivalence fails: a next-to-shortest path may have weight less than $d_G(s,t)+1$.) The $k$ shortest \emph{walks} problem is solvable very efficiently~\cite{eppstein98}, but walks may repeat vertices and are not useful here.

\subsection{Organization}

Section~\ref{sec:prelim} introduces the shortest-path DAG and the completion operation. Section~\ref{sec:fixedA} shows that the fixed-deviation-point subproblem is NP-complete. Section~\ref{sec:structure} develops the local replacement structure of a global optimum, culminating in the uncrossing lemma and the Local Replacement Theorem. Section~\ref{sec:algorithm} defines the representatives, proves that they cover an optimal solution, and gives the dynamic program and the final running-time analysis. Section~\ref{sec:conclusion} concludes with open questions.

\section{The shortest-path DAG and completions}\label{sec:prelim}

Write
\[
d(v)=d_G(s,v),\qquad d_t(v)=d_G(v,t),\qquad D=d(t).
\]
If $D=\infty$ there is no solution. A vertex $v$ is \emph{straight} (the term follows~\cite{bss21,cwz}) if
\[
d(v)+d_t(v)=D,
\]
and an edge $uv$ is \emph{tight} if
\[
d(u)+w(uv)+d_t(v)=D.
\]
We may first delete all vertices that are not reachable from $s$ or cannot reach $t$; in the nontrivial case this does not change the answer, and afterwards $m=\Omega(n)$. Let $H$ be the subgraph formed by all straight vertices and all tight edges. Along every tight edge,
\[
d(v)=d(u)+w(uv)>d(u),
\]
so $H$ is a DAG. A path in $H$ is called a \emph{tight path}. Write $h=|V(H)|$ and $q=|E(H)|$. We write $u\Hreach v$ if $H$ contains a path from $u$ to $v$, allowing $u=v$.

All paths are simple; paths of length zero are allowed. Intersections and unions of paths refer to vertex sets. $W[u,v]$ denotes the subpath of $W$ from $u$ to $v$; a round bracket excludes the corresponding endpoint. $W^\circ$ denotes the set of internal vertices of $W$, i.e.\ $W$ with both endpoints removed.

\begin{lemma}\label{lem:shortest-tight}
If $v$ is a straight vertex, then every shortest $s\to v$ path and every shortest $v\to t$ path is a tight path. In particular, a simple $s$--$t$ path is a shortest path if and only if it is a tight path.
\end{lemma}

\begin{proof}
Let $I$ be a shortest $s\to v$ path and $J$ a shortest $v\to t$ path. The walk $I\circ J$ has length $D$. If it repeated a vertex, removing a cycle of positive length would give an $s$--$t$ path of length less than $D$, a contradiction. Hence $I\circ J$ is a shortest simple path and all its edges are tight. The same argument applies to $J$. The last sentence follows immediately.
\end{proof}

For straight vertices $A,B$, a pair of vertex-disjoint tight paths
\[
P\colon s\to A,\qquad Q\colon B\to t
\]
is called a \emph{feasible pair} for $(A,B)$.

Given a tight path $P\colon s\to A$ with $A\ne t$, delete from $G$ all vertices of $P\setminus\{A\}$ and then all tight out-edges of $A$; call the resulting graph $G_P$. Define
\[
C(P)=d(A)+d_{G_P}(A,t),
\]
where the value $\infty$ is allowed.

\begin{lemma}[Completion]\label{lem:completion}
If $C(P)<\infty$, then $P$ followed by a shortest $A\to t$ path in $G_P$ is a non-shortest simple $s$--$t$ path of length $C(P)$.
\end{lemma}

\begin{proof}
A shortest path in the residual graph can be taken simple, and it avoids $P\setminus\{A\}$, so the concatenation is simple. Its first continuation edge is not tight, hence the whole path is not a shortest path.
\end{proof}

\section{The subproblem with a fixed deviation point}\label{sec:fixedA}

If a next-to-shortest path exists, let $\OPT$ denote its length. For a straight vertex $A\ne t$ define
\[
F(A)=\min\{C(P):P\text{ is a tight path }s\to A\}.
\]

\begin{lemma}\label{lem:fixedA-decomp}
If a next-to-shortest path exists, then
\[
\OPT=\min_{A\ne t}F(A).
\]
\end{lemma}

\begin{proof}
By Lemma~\ref{lem:completion}, every finite $C(P)$ is the length of a valid non-shortest path, so $F(A)\ge\OPT$. Conversely, take a next-to-shortest path $W^*$ and let $L\colon s\to A$ be its longest tight prefix. Since $W^*$ is simple, its suffix from $A$ does not touch $L\setminus\{A\}$; since $L$ is the longest tight prefix, the first edge leaving $A$ is not tight and is therefore not deleted in the construction of $G_L$. Hence this suffix lies in $G_L$ and $C(L)\le\OPT$. Thus $F(A)\le\OPT$ for some $A$.
\end{proof}

Nevertheless, computing a single value $F(A)$ is already hard.

\begin{proposition}\label{prop:fixedA-hard}
Given a positively weighted digraph $G$ with integer weights, vertices $s,t$ and a straight vertex $A\ne t$, deciding whether there exists a tight path $P\colon s\to A$ with $C(P)<\infty$ is NP-complete. The edge weights need only take the value $1$ and one integer of polynomial size.
\end{proposition}

\begin{proof}
We reduce from the problem 2D1SP of Eilam-Tzoreff~\cite{eilam98}: given an undirected graph $G_0$ with unit edge lengths and four distinct vertices $s_1,t_1,s_2,t_2$, decide whether there exist vertex-disjoint paths $P_1\colon s_1\to t_1$ and $P_2\colon s_2\to t_2$ such that $P_1$ is a shortest path. Its vertex-disjoint version on undirected graphs with unit edge lengths is NP-complete by Claim~1 in Section~2 of~\cite{eilam98}. If $s_1$ and $t_1$ are disconnected, the instance is trivially negative; in that case we output any fixed negative instance satisfying the promise of the proposition. Assume henceforth $d_{G_0}(s_1,t_1)<\infty$.

Replace each edge of $G_0$ by two oppositely directed arcs of unit weight, and delete all out-arcs of $t_1$. Let $s=s_1$ and $A=t_1$, add a new vertex $t$ and three arcs
\[
A\to t\quad(1),\qquad A\to s_2\quad(M),\qquad t_2\to t\quad(M),
\]
where $M=|V(G_0)|+2$. Let $d_0=d_{G_0}(s_1,t_1)$. Every $s$--$t$ path that uses an arc of weight $M$ is longer than $d_0+1$, so the shortest $s$--$t$ paths of the new graph are exactly ``a shortest $s_1\to t_1$ path in $G_0$ followed by $A\to t$.'' Thus tight paths $s\to A$ correspond exactly to shortest $s_1\to t_1$ paths in $G_0$.

Fix such a tight path $P$, corresponding to $P_1$. In $G_P$ the arc $A\to t$ is deleted, the only possible remaining out-arc of $A$ is $A\to s_2$, and the only possible remaining arc entering the new vertex $t$ is $t_2\to t$; at least one of them is absent if $s_2$ or $t_2$ lies on $P_1$, in which case $C(P)=\infty$. Otherwise a simple $A\to t$ path of $G_P$ starts with $A\to s_2$, ends with $t_2\to t$, and cannot return to $A$, so its middle part is an $s_2\to t_2$ path of $G_0-V(P_1)$; conversely any such path gives a completion. Hence $C(P)<\infty$ if and only if $G_0-V(P_1)$ contains an $s_2\to t_2$ path, i.e.\ a path $P_2$ vertex-disjoint from $P_1$. This proves the reduction. Membership in NP is clear: a tight path $P$ and an $A\to t$ path of $G_P$ form a certificate that can be checked with two shortest-path computations.
\end{proof}

Consequently, unless $\mathrm{P}=\mathrm{NP}$, no polynomial-time algorithm evaluates $F(A)$ for an arbitrary specified $A$ (although the decomposition of Lemma~\ref{lem:fixedA-decomp} is of course valid). Below we use the structure of a global optimum to construct a polynomial-size family of representatives.

\section{Local replacement structure of a global optimum}\label{sec:structure}

Fix a globally optimal non-shortest simple path
\[
W^*=L_0\circ K\circ R_0,
\]
where $L_0\colon s\to A$ is the longest tight prefix, $R_0\colon B\to t$ is the longest tight suffix, and $K=W^*[A,B]$. Neither the first nor the last edge of $K$ is tight. This decomposition follows the optimal-middle-segment view of CWZ; in their layered setting $K$ corresponds to the segment $P_0^*$ between the first and the last back-edge. Here the prefix and the suffix are defined by tight edges of the input graph. Note that in a general graph a forward edge in the sense of CWZ, defined by $d(u)+w(uv)=d(v)$, need not lie on any shortest $s$--$t$ path, so the two notions do not coincide. Let
\[
S=K^\circ\cap V(H).
\]
The non-straight internal vertices of $K$ lie on no tight path. Relative to this fixed $K$, a feasible pair $(P,Q)$ is called \emph{clean} if
\[
(P\cup Q)\cap S=\varnothing,
\]
and \emph{dirty} otherwise. The original pair $(L_0,R_0)$ is clean.

\subsection*{Middle-segment rigidity and barriers}

\begin{lemma}[Middle-segment rigidity]\label{lem:rigidity}
Let $u,v\in S\cup\{A,B\}$ with $u$ preceding $v$ on $K$ and $(u,v)\ne(A,B)$. If there exist vertex-disjoint tight paths
\[
\sigma\colon s\to u,\qquad \tau\colon v\to t
\]
with
\[
\sigma\cap K[u,v]=\{u\},\qquad
\tau\cap K[u,v]=\{v\},
\]
then $K[u,v]$ is a tight path.
\end{lemma}

\begin{proof}
$\sigma\circ K[u,v]\circ\tau$ is a simple $s$--$t$ path. Replacing the $s\to u$ prefix of $W^*$ by $\sigma$ does not increase the length; if $u\ne A$, the original prefix contains the first, non-tight edge of $K$ while $u$ is straight, so by Lemma~\ref{lem:shortest-tight} it is not a shortest $s\to u$ path and the replacement strictly shortens. The same holds for the suffix. Since $(u,v)\ne(A,B)$, at least one side strictly shortens.

If $K[u,v]$ still contained a non-tight edge, the spliced simple path would be non-shortest yet shorter than $W^*$, a contradiction. Hence the spliced path must be a shortest path, and every edge of $K[u,v]$ is tight.
\end{proof}

\begin{corollary}[Barrier property]\label{cor:barrier}
Let $(L,R)$ be any clean feasible pair for $(A,B)$. Then:
\begin{enumerate}[label=\textup{(\roman*)}]
\item Every tight path from $s$ to a vertex of $S$ passes through $R$; every tight path from a vertex of $S$ to $t$ passes through $L$.
\item For every $u\in S$ there is a tight path $B\to u\to A$. More generally, every tight path from a vertex of $L$ to $S$ passes through $R$, and every tight path from $S$ to a vertex of $R$ passes through $L$.
\item For every feasible pair $(P,Q)$ for $(A,B)$,
\[
P\cap S=\varnothing\quad\Longleftrightarrow\quad Q\cap S=\varnothing.
\]
\end{enumerate}
\end{corollary}

\begin{proof}
We first prove (i). Suppose a tight path $T\colon s\to u$ with $u\in S$ avoids $R$, and let $u_0$ be the first vertex of $T$ that belongs to $S$. Then $T[s,u_0]$ and $R$ splice $K[u_0,B]$ into a simple path. By Lemma~\ref{lem:rigidity}, $K[u_0,B]$ would be tight, but it contains the last edge of $K$, a contradiction. The other statement is symmetric.

For (ii), take any tight path $s\to u$. By (i) it meets $R$; the prefix of $R$ from $B$ to that intersection followed by the remainder of the path is a tight path $B\to u$. The path $u\to A$ is symmetric. The remaining two statements follow directly from (i) after prepending $L$ or appending $R$ to the path.

For (iii), suppose $P\cap S=\varnothing$ but $Q\cap S\ne\varnothing$, and let $u$ be the earliest vertex of $Q\cap S$ along $K$. Then $P$ and $Q[u,t]$ splice $K[A,u]$ into a simple path, so by Lemma~\ref{lem:rigidity} this subsegment would be tight; but it contains the first edge of $K$. For the converse direction take the last vertex of $P\cap S$ along $K$ and use $K[u,B]$.
\end{proof}

\subsection*{The boundary edge and local replacement}

Assume from now on that $S\ne\varnothing$. Let
\[
\mathcal U=\Reach_H(S)=\{v:\exists u\in S,\ u\Hreach v\}.
\]
$\mathcal U$ is closed under tight reachability. By Corollary~\ref{cor:barrier}(ii), $A\in\mathcal U$. Every $u\in S$ satisfies $s\Hreach u$, and $H$ is acyclic, so $s\notin\mathcal U$. Hence along $L_0\colon s\to A$ there is a unique edge entering $\mathcal U$ from outside; call it $x\to c$. Then
\[
x\notin\mathcal U,\qquad c\in\mathcal U,\qquad L_0[c,A]\subseteq\mathcal U.
\]
We call $x\to c$ the \emph{boundary edge} induced by this optimal solution.

Fix any topological numbering of $H$, i.e.\ a bijection
\[
r\colon V(H)\stackrel{\sim}{\longrightarrow}\{0,1,\ldots,h-1\}
\]
such that every tight edge goes from a smaller to a larger number.

\begin{theorem}[Local Replacement Theorem]\label{thm:replacement}
Let $(P,Q)$ be a feasible pair for $(A,B)$ with $x\in P$, and let $y$ be the successor of $x$ on $P$. If
\[
r(y)\ge r(c),
\]
then $(P,Q)$ is clean. Consequently $P\circ K\circ Q$ is a next-to-shortest path of length $\OPT$, and $C(P)=\OPT$.
\end{theorem}

In the algorithm this theorem takes the place of the two-edge certificate (the Key Lemma) of CWZ: it does not constrain the second path to pass through a prescribed edge, but guarantees compatibility through the choice of the next hop on the first path.

We first present an uncrossing operation and then prove the theorem by a strict descent of the last intersection. Throughout, $W^*,K,S,x,c$ are fixed. A clean feasible pair for $(A,B)$ whose first path contains the edge $xc$ is called a \emph{reference pair}. The original pair $(L_0,R_0)$ is one. For every reference pair $(L,R)$, $c$ is still the first vertex at which $L$ enters $\mathcal U$: otherwise, following $L$ from an earlier intersection to $x$ would give $x\in\mathcal U$.

The common vertices of two tight paths appear in the same order on both, since otherwise $H$ would contain a directed cycle. Hence the ``last common vertex'' below does not depend on which path's order is used.

\subsection*{Uncrossing and strict descent of the last intersection}

\begin{lemma}[Uncrossing]\label{lem:uncrossing}
Let $(P,Q)$ be a dirty feasible pair for $(A,B)$ with $x\in P$ such that the successor $y$ of $x$ on $P$ satisfies $r(y)\ge r(c)$. For every reference pair $(L,R)$, the set $L\cap Q$ is nonempty; let $a$ be its last common vertex. Then there exists another reference pair $(\bar L,\bar R)$ such that $\bar L\cap Q$ is also nonempty and its last common vertex $\bar a$ strictly precedes $a$ on $Q$. In particular,
\begin{equation}\label{eq:last-intersection}
r(\bar a)<r(a).
\end{equation}
\end{lemma}

\begin{proof}
\textbf{Constructing the barrier.}
By Corollary~\ref{cor:barrier}(iii), $Q\cap S\ne\varnothing$. By (i), the suffix of $Q$ from any vertex of $S$ to $t$ passes through $L$, so $a$ exists and is reachable from $S$. Moreover $A\notin Q$, hence
\[
a\in L[c,A),\qquad Q[a,t]\cap S=\varnothing.
\]
The second relation holds because if this tail met $S$ again, it would have to pass through $L$ once more.

Let $z$ be the earliest vertex along $K$ with $z\in S$ and $z\Hreach c$, and take a tight path $\pi\colon z\to c$. By the boundary property of $c$, $\pi\cap L=\{c\}$. If $\pi$ met $R$, appending the suffix of $R$ at the intersection would give a tight path $S\to t$ avoiding $L$, violating Corollary~\ref{cor:barrier}(i). Hence $\pi\cap R=\varnothing$.

Set
\[
\tau=\pi\circ L[c,a]\circ Q[a,t].
\]
It is a directed walk in $H$, hence a simple tight path. By the minimality of $z$, the cleanness of the reference pair and $Q[a,t]\cap S=\varnothing$,
\[
\tau\cap K[A,z]=\{z\}.
\]
Here $A\notin\tau$, because $\pi\cap L=\{c\}$, $c\ne A$, $a\ne A$ and $A\notin Q$.

We claim that $\tau$ meets every tight path $s\to A$. Otherwise take a tight path $\sigma\colon s\to A$ avoiding $\tau$, and let $u$ be the first vertex at which it meets $K[A,z]$. Such $u$ exists, and $u\ne z$. Then
\[
\sigma[s,u]\circ K[u,z]\circ\tau
\]
is a simple path. $K[u,z]$ is a nontrivial proper subsegment, so by middle-segment rigidity it must be tight. If $u=A$ it contains the first, non-tight edge of $K$, a contradiction; otherwise $u\in S$ precedes $z$ along $K$ and $u\Hreach z\Hreach c$, contradicting the choice of $z$. Hence $\tau$ is an unavoidable barrier.

\medskip\noindent\textbf{Locating the splice points.}
We first show $P\cap\pi\subseteq\{c\}$. $P[s,x]$ does not meet $\pi$, for otherwise $x\in\mathcal U$. The vertices of $P$ after $x$ have numbers at least $r(y)\ge r(c)$, whereas the vertices of $\pi$ other than $c$ have numbers below $r(c)$. Together with $P\cap Q=\varnothing$ we obtain
\[
\varnothing\ne P\cap\tau\subseteq L[c,a).
\]
Let $b$ be the last vertex along $P$ belonging to $\tau$. On the other hand, $R$ avoids $\pi$ and $L$, but $t\in R\cap\tau$. Let $q$ be the first vertex along $R$ belonging to $\tau$. Then
\[
\begin{aligned}
&b\in L[c,a), &&P[b,A]\cap\tau=\{b\},\\
&q\in Q(a,t], &&R[B,q]\cap\tau=\{q\}.
\end{aligned}
\]
Since $P\cap Q=\varnothing$, we have $b\ne q$.

Take any tight path $I\colon s\to B$. Corollary~\ref{cor:barrier}(ii) gives $B\Hreach z$ and $B\ne z$. If $I$ met $\tau$ at a vertex $v$, then $B\Hreach z\Hreach v\Hreach B$ would be a directed cycle. Hence $I\cap\tau=\varnothing$. If $R[B,q]$ and $P[b,A]$ met at a vertex $v$, then $v\notin\tau$ and
\[
I\circ R[B,v]\circ P[v,A]
\]
would be a tight path $s\to A$ avoiding $\tau$, a contradiction. Therefore
\begin{equation}\label{eq:splice-disjoint}
R[B,q]\cap P[b,A]=\varnothing.
\end{equation}

\Needspace{7\baselineskip}
\medskip\noindent\textbf{Exchanging the tails.}
First set
\[
\bar R=R[B,q]\circ Q[q,t].
\]
Since $q$ is the first contact of $R$ with $\tau$ and $Q[q,t]\subseteq\tau$, the two pieces meet only at $q$. Moreover $q$ lies after $a$ on $Q$, so $\bar R$ avoids $L$. Thus $(L,\bar R)$ is feasible. $L$ is clean, and Corollary~\ref{cor:barrier}(iii) shows that $\bar R$ is clean as well.

Next set
\[
\bar L=L[s,b]\circ P[b,A].
\]
If the two pieces shared a vertex other than $b$, $H$ would contain a directed cycle, so $\bar L$ is a simple tight path. By~\eqref{eq:splice-disjoint} and $P\cap Q=\varnothing$, $P[b,A]$ avoids $\bar R$. Hence $(\bar L,\bar R)$ is feasible; since $\bar R$ is clean, Corollary~\ref{cor:barrier}(iii) shows that $\bar L$ is clean too. Also $b\in L[c,A)$, so $\bar L$ still contains the edge $xc$.

Finally, $Q$ meets $S$, so the barrier property gives $\bar L\cap Q\ne\varnothing$. At the same time
\[
\bar L\cap Q=L[s,b]\cap Q\subseteq L[s,b).
\]
Every new intersection is an old one and lies before $b$, and $b$ lies before $a$. Hence the last intersection $\bar a$ strictly precedes $a$ on $Q$, which gives~\eqref{eq:last-intersection}.
\end{proof}

\subsection*{Proof of the Local Replacement Theorem}

\begin{proof}[Proof of Theorem~\ref{thm:replacement}]
Suppose for contradiction that $(P,Q)$ is dirty. Starting from the reference pair $(L_0,R_0)$, apply Lemma~\ref{lem:uncrossing} repeatedly. Each resulting pair is again clean and keeps the edge $xc$; the last intersection of its first path with $Q$ always exists and moves to a strictly earlier position along $Q$. If these intersections are $a_0,a_1,\ldots$, then
\[
r(a_0)>r(a_1)>r(a_2)>\cdots\ge0,
\]
contradicting the finiteness of $Q$. Hence $(P,Q)$ is clean.

Tight paths do not pass through the non-straight internal vertices of $K$; cleanness excludes the straight internal vertices, and together with $P\cap Q=\varnothing$ the path $P\circ K\circ Q$ is a simple non-shortest path with
\[
w(P\circ K\circ Q)=d(A)+w(K)+d_t(B)=\OPT.
\]
$K\circ Q$ is also an available completion in $G_P$, so $C(P)\le\OPT$. Lemma~\ref{lem:completion} gives the reverse inequality, hence $C(P)=\OPT$.
\end{proof}

\begin{corollary}\label{cor:edgecert}
Every feasible pair for $(A,B)$ whose first path contains the edge $xc$ can be combined with the same middle segment $K$ into an optimal next-to-shortest path, and satisfies $C(P)=\OPT$.
\end{corollary}

\begin{proof}
In this case the next hop is $c$ itself, so $r(y)=r(c)$, and Theorem~\ref{thm:replacement} applies directly.
\end{proof}

\Needspace{12\baselineskip}
\section{Representatives and the algorithm}\label{sec:algorithm}

\subsection*{Latest feasible next hop}

We keep the topological numbering $r$ fixed in Section~\ref{sec:structure}. For straight vertices $A,B,x$ with $A\ne t$, define the set of feasible next hops
\[
Y(A,B,x)=\{y:\exists\text{ a feasible pair }(P,Q)\text{ for }(A,B)\text{ with }xy\in E(P)\}.
\]
If $Y(A,B,x)\ne\varnothing$, take the vertex $y$ of maximum topological number $r(y)$ in it and keep the first path of one feasible pair realizing $xy$. If $x=A$ and a feasible pair for $(A,B)$ exists, additionally keep the first path of an arbitrary such pair. All kept first paths form the family of \emph{representatives} $\pathset$; clearly
\[
|\pathset|=O(n^3).
\]

\begin{theorem}[Coverage]\label{thm:coverage}
If a next-to-shortest path exists, then
\[
\min_{P\in\pathset}C(P)=\OPT.
\]
\end{theorem}

\begin{proof}
By Lemma~\ref{lem:completion}, every finite $C(P)$ is the length of a valid non-shortest path, so the left-hand side is at least $\OPT$.

Fix the global optimal decomposition $W^*=L_0KR_0$ of Section~\ref{sec:structure}. If $S=\varnothing$, every feasible pair for $(A,B)$ is clean, so any representative with $x=A$ can be used.

Suppose $S\ne\varnothing$ and let $x\to c$ be the boundary edge. The original clean pair $(L_0,R_0)$ shows that
\[
c\in Y(A,B,x).
\]
Hence the next hop $y$ chosen by the algorithm for this triple satisfies $r(y)\ge r(c)$. By Theorem~\ref{thm:replacement}, its representative $P$ satisfies $C(P)=\OPT$.
\end{proof}

\subsection*{A two-path dynamic program with rewards}

The actual algorithm works directly with the topological numbering $r$ fixed in Section~\ref{sec:structure} and maximizes $r(y)$ within each $Y(A,B,x)$.

We use the standard ``two tokens on a DAG'' dynamic program; see Fortune, Hopcroft and Wyllie~\cite{fhw80}. Fix a straight vertex $A\ne t$. Assign to the tight edges of the first path an arbitrary additive reward $\rho\colon E(H)\to\R$. Define
\[
T_A^\rho(u,v)=\max\{\rho(P):
P\colon u\to A,\ Q\colon v\to t\text{ are vertex-disjoint tight paths}\},
\]
with value $-\infty$ if infeasible, using the convention $\max\varnothing=-\infty$.

\begin{lemma}[Two-path DP with rewards]\label{lem:dp}
$T_A^\rho$ can be computed in $O(hq)$ time and $O(h^2)$ space.
\end{lemma}

\begin{proof}
The boundary conditions are
\[
T_A^\rho(A,t)=0,\qquad T_A^\rho(u,u)=-\infty.
\]
Apart from the terminal state, if $u=A$ we advance the second path; if $v=t$ we advance the first path; otherwise we advance the one whose current vertex has the smaller topological number.

If the first path is advanced,
\[
T_A^\rho(u,v)=
\max_{uy\in E(H),\ y\ne v}
\bigl(\rho(uy)+T_A^\rho(y,v)\bigr).
\]
If the second path is advanced,
\[
T_A^\rho(u,v)=
\max_{vz\in E(H),\ z\ne u}
T_A^\rho(u,z).
\]
We argue by induction in reverse topological order on both coordinates, and only discuss the first case. Deleting the first vertex $u$ of the first path yields a disjoint pair of suffixes starting at $(y,v)$. Conversely, when $u$ is prepended, its topological number is smaller than that of $y$, and smaller than that of $v$ whenever this is needed; neither suffix can return to $u$. Hence the recurrence is both necessary and sufficient. Since the reward is additive, the maximum satisfies the same recurrence.

There are $O(h^2)$ states; each tight edge is combined with at most $h$ positions of the other coordinate, so the total number of successor checks is $O(hq)$. Every vertex of $H$ lies on some tight $s$--$t$ path, so $q\ge h-1$ and the total time is $O(h^2+hq)=O(hq)$.
\end{proof}

For each straight vertex $x$, take the local reward
\[
\rho_x(uv)=
\begin{cases}
r(v)+1,&u=x,\\
0,&u\ne x.
\end{cases}
\]
A tight path leaves $x$ at most once, so if the first path passes through $x$ and continues, its total reward is exactly the number of the next hop plus one; otherwise it is zero. Thus, for fixed $(A,x)$, the table $T_A^{\rho_x}(s,B)$ yields the representatives for all $B$ simultaneously. If a feasible first path through $x$ exists, the maximum corresponds to the feasible next hop of maximum number; if not, a zero-reward representative is merely an additional valid candidate.

\paragraph{Algorithm.}
Compute $d$, $d_t$, the tight graph $H$ and the topological numbering $r$, and set $\mathrm{best}=\infty$. Then, for every straight vertex $A\ne t$ and every straight vertex $x$, compute the table $T_A^{\rho_x}$. For every straight vertex $B\ne s,A$ with $T_A^{\rho_x}(s,B)>-\infty$, recover a feasible pair attaining the maximum, take its first path $P$, compute $C(P)$ and update $\mathrm{best}$ with it. If in the end $\mathrm{best}=\infty$, report that no next-to-shortest path exists; otherwise output a path attaining this value.

\begin{theorem}\label{thm:complexity}
The algorithm above is correct and can be implemented in
\[
O\bigl(n^3(m+n\log n)\bigr)
\]
time with working space $O(n+m+h^2)$. If all edge weights equal $1$, the running time is $O(n^3m)$.
\end{theorem}

\begin{proof}
Theorem~\ref{thm:coverage} guarantees that the candidates contain the optimal value, and Lemma~\ref{lem:completion} guarantees that every finite candidate is valid; hence the algorithm is correct.

We first compute $d$ and $d_t$ by two runs of Dijkstra's algorithm, stop if $d(t)=\infty$, and delete the vertices not reachable from $s$ or not reaching $t$. One table $T_A^{\rho_x}$ takes $O(hq)$ time; together with each finite nonterminal entry we store a successor attaining the maximum, so that a representative is recovered by following these pointers from $(s,B)$ until the terminal entry $(A,t)$; they advance one coordinate at a time and visit at most $h$ entries. There are $O(h^2)$ pairs $(A,x)$, so all tables take $O(h^3q)$ time. Each triple $(A,B,x)$ recovers at most one representative, for $O(h^3)$ in total. Each recovery takes $O(h)$ time; building $G_P$ takes $O(n+m)$ time; and each completion is one run of Dijkstra's algorithm, taking $O(m+n\log n)$ time with Fibonacci heaps~\cite{ft87}. The total time is therefore
\[
O\bigl(h^3(m+n\log n)\bigr)
\subseteq O\bigl(n^3(m+n\log n)\bigr).
\]
Processing the pairs $(A,x)$ one at a time, and keeping only the current best path, keeps the working space at $O(n+m+h^2)$. On unweighted graphs, Dijkstra's algorithm is replaced by breadth-first search; after preprocessing $m=\Omega(n)$, so the time is $O(n^3m)$. As usual for real weights, the bound counts arithmetic operations and comparisons on weights.
\end{proof}

\section{Concluding remarks}\label{sec:conclusion}

We have shown that the next-to-shortest path problem on positively weighted digraphs can be solved in $O(n^3(m+n\log n))$ time, improving the $O(n^4m^3\log n)$ bound of Chen, Wein and Zhang~\cite{cwz}. The algorithm consists of $O(n^2)$ runs of a two-token DAG dynamic program followed by $O(n^3)$ ordinary shortest-path computations, and its correctness rests on a single structural statement about a fixed optimal solution, the Local Replacement Theorem.

Several questions remain. Two stages contribute to the running time: computing $O(h^2)$ tables, each in $O(hq)$ time, and performing $O(h^3)$ completions, each in $O(m+n\log n)$ time. Improving only one stage need not improve the overall exponent; reducing the number of representatives, sharing tables across different $x$, or sharing work between completions with the same prefix, are natural directions. It would also be interesting to know whether the fixed-deviation-point hardness of Proposition~\ref{prop:fixedA-hard} can be circumvented in special graph classes, giving near-linear algorithms in the spirit of the undirected results~\cite{wu13,wgw12}. Finally, CWZ raise the $k$-longest detour problem on general unweighted digraphs; whether the uncrossing technique of Section~\ref{sec:structure} extends to detours with at least $d_G(s,t)+k$ edges for fixed $k>1$ is open.

\paragraph{Acknowledgments.}
AI tools were used to assist with the research, the program-based checking and the writing of this paper. The author is responsible for the verification of all mathematical claims.

\end{document}